\documentclass[journal,12pt,onecolumn,draftcls]{IEEEtran}
\usepackage{float}
\usepackage{comment}
\usepackage[table]{xcolor}
\usepackage{cite}      
\usepackage{graphicx}  
\usepackage{psfrag}    
\usepackage{placeins}
\usepackage{lineno}
\usepackage{url}
\usepackage{subcaption}
\usepackage{caption}
\usepackage{times}
\usepackage{textcomp}
\usepackage{amsmath} 
\usepackage{amsfonts,amsthm,amssymb,mathtools}
\usepackage{balance}
\usepackage[normalem]{ulem}
\usepackage{xcolor}

\usepackage{cuted}
\usepackage{amsmath,bm}
\usepackage{hyperref}
\usepackage{cleveref}

\DeclareMathOperator{\E}{\mathbb{E}}
\DeclareMathOperator{\Pro}{\prod_{c \in \Phi \cap \mathbf{A_c}}}
\DeclareMathOperator{\rcst}{\sigma_{t_{avg}}}
\DeclareMathOperator{\rcsc}{\sigma_{c_{avg}}}
\DeclareMathOperator{\Ei}{\underset{\mathbf{\sigma_c},c \in \Phi \cap \mathbf{A_c}}{\E}}
\DeclareMathOperator{\Eii}{\underset{\sigma_c,\mathbf{A_c}}{\E}}
\DeclareMathOperator{\Su}{\sum_{c \in \Phi \cap \mathbf{A_c}}}

\DeclareMathOperator{\hxt}{\mathcal{H}(\kappa)}
\DeclareMathOperator{\hxc}{\mathcal{H}(\kappa_c)}

\DeclareMathOperator{\Pdc}{\mathcal{P}_{DC}^{Bi}}
\DeclareMathOperator{\thetat}{\mathbf{\theta_t}}
\DeclareMathOperator{\thetac}{\mathbf{\theta_c}}
\DeclareMathOperator{\longvar}{\rcst P_{tx}A_0\hxt}
\DeclareMathOperator{\twotheta}{\Delta \theta_{tx}\Delta \theta_{rx}}

\newtheorem{theorem}{Theorem}[]
\newtheorem{corollary}{Corollary}[theorem]

\usepackage{multicol}
\usepackage{xr-hyper}

\newtheorem{definition}{Definition}

\begin{document}
\title{Estimation of Bistatic Radar Detection Performance Under Discrete Clutter Conditions Using Stochastic Geometry}
\author{Shobha~Sundar~Ram and Gourab~Ghatak\\ Indraprastha Institute of Information Technology Delhi, India}
\maketitle

\begin{abstract}
We propose a metric called the bistatic radar detection coverage probability to evaluate the detection performance of a bistatic radar under discrete clutter conditions. 
Such conditions are commonly encountered in indoor and outdoor environments where passive radars receivers are deployed with opportunistic illuminators. Backscatter and multipath from the radar environment give rise to ghost targets and point clutter responses in the radar signatures resulting in deterioration in the detection performance. In our work, we model the clutter points as a Poisson point process to account for the diversity in their number and spatial distribution. Using stochastic geometry formulations we provide an analytical framework to estimate the probability that the signal to clutter and noise ratio from a target at any particular position in the bistatic radar plane is above a predefined threshold. Using the metric, we derive key radar system perspectives regarding the radar performance under noise and clutter limited conditions; the range at which the bistatic radar framework can be approximated to a monostatic framework; and the optimal radar transmitted power and bandwidth. Our theoretical results are experimentally validated with Monte Carlo simulations.
\end{abstract}
\providecommand{\keywords}[1]{\textbf{\textit{Keywords--}}#1}

\section{Introduction}
Historically, bistatic radars were researched for military applications - for tracking missiles, for spotting certain types of targets whose radar cross-section or Doppler information is enhanced in a bistatic/multi-static configuration or for detecting targets in a forward scattering scenario - when they interrupt the direct path between a radar transmitter and receiver~\cite{skolnik1961analysis, griffiths2003different, willis2005bistatic, davis2007advances}. However, there are some significant disadvantages in bistatic systems. First, the synchronization between the radar transmitter and receiver is challenging. Second, the view of the bistatic radar may be more limited than a monostatic radar since a target has to be visible to both the transmitting and the receiving antennas. Finally, the localization and Doppler estimation algorithms are significantly more complex in a bistatic scenario when compared to monostatic radar \cite{willis2005bistatic}. Despite these limitations, recently, there has been a revival of interest in the research and development of bistatic radar specifically in the context of microwave and millimeter wave integrated sensing and communication systems (ISAC). For example, radar sensing using wireless access points as opportunistic illuminators, with passive, cheap and undetectable receivers, have been researched for detecting, tracking and monitoring human activities in indoor environments \cite{falcone2012potentialities,li2020passive}. Another example is millimeter based vehicle-to-everything (V2X) communication systems that are increasingly incorporating sensing functionality within their communication protocols \cite{ali2020passive}. There are several advantages to these ISAC systems: first, dual functionality is realized with common hardware and spectrum resources \cite{duggal2019micro,duggal2020doppler}. Second, in some cases, the incorporation of the localization functionality facilitates improvement in communication metrics.
 
Indoor wireless and localization systems encounter significant discrete clutter scatterers such as furniture and walls with strong specular backscatter in the radar signatures. Further, they give rise to multipath off target backscatter which manifests as ghost targets in the radar signatures \cite{vishwakarma2020mitigation}. Similarly, in outdoor automotive scenarios, reflections from roads, railings and surrounding environments give rise to significant point clutter responses in the radar signatures. While there has been statistical analysis of bistatic radar detection performance due to surface clutter \cite{al2011statistical}, there is little research into quantifying the performance due to discrete clutter since there can be considerable variation in the environmental conditions.
In this work, we propose a metric called bistatic radar detection coverage probability ($\Pdc$) to quantify the radar detection performance based on the signal to clutter and noise ratio (SCNR) using stochastic geometry (SG) principles. $\Pdc$ specifies the fraction of targets in the bistatic region-of-interest which have SCNR above a predefined threshold. In \cite{ram2020estimating,ram2021optimization} we used SG to derive the monostatic radar detection coverage probability to quantify the performance of a monostatic radar under discrete clutter conditions. The metric accounted for wide variation in the radar, target and clutter conditions and provided system perspectives on radar deployment and optimization of bandwidth, pulse width, maximum power and gains of  transmitting and receiving antennas. In this work, we use SG principles to model the distribution of the discrete clutter scatterers as a Poisson point process (PPP) similar to \cite{chen2012integrated}. The PPP is a natural choice since it is a completely random process  where each clutter point is stochastically independent to the location of all the other clutter points. Thus it accounts for the diversity in the number and spatial distribution of the clutter scatterers around the radar. With this model, we consider each radar deployment as an independent instance of an underlying spatial stochastic process. The PPP assumption leads to elegant expressions of the SCNR and enables the radar operator to dimension the deployments and derive useful insights not only throughout different spatial locations in a particular deployment, but also across different deployments with modified parameters. 

We use the classical bistatic radar geometry framework proposed by \cite{jackson86} where the bistatic range is identified as the geometric mean ($\kappa$) of the one-way propagation distances between the radar transmitter and receiver with respect to the target. An important important parameter in bistatic radar is the baseline length ($L$) between the transmitter and receiver. In our work, we derive the SCNR for scenarios where $L\leq 2\kappa$ by considering the returns from a single target and from clutter scatterers within the same clutter resolution cell as the target. We consider two types of clutter resolution cells - beamwidth limited clutter resolution cell which is the region in the bistatic plane within the main lobes of both the transmitting and receiving antennas; and the range limited resolution cell which is the region defined by the beamwidth of either antenna and the range resolution bin. Based on the SCNR, our main contribution in this paper is an analytical framework that quantifies the radar performance as a function of radar parameters (such as beamwidth, transmitted power, bandwidth), clutter parameters (clutter distribution, and clutter cross-section) and target parameters (target cross-section). Second, we provide key system insights into radar performance under noise and clutter limited conditions; the maximum power at which detection performance saturates; the optimal radar bandwidth for best detection performance; and the bistatic radar range as a function of baseline length, at which the bistatic radar approximates to monostatic radar. We experimentally validate the analytical results with Monte Carlo simulations. The paper is organized as follows. In the following section, we present the bistatic radar geometry followed by theorems to compute $\Pdc$ for both beamwidth limited and range limited radar resolution cells along with their corollaries. Then in Section.\ref{sec:Results}, we present the radar detection performance as a function of radar, target and clutter parameters. Finally, we conclude the paper with a summary of the insights gained by the study in Section.\ref{sec:Conclusion}

\emph{Notation:} We have used the following notation in this paper. Random variables are indicated by bold face fonts while regular parameters and realizations of random variables are written using regular font. 
\section{Theory}
\label{sec:Theory}
In this paper, we follow the North referenced bistatic geometry configuration popularized by Jackson \cite{jackson86}. We consider a two-dimensional Cartesian coordinate space with the bistatic radar transmitter and receiver at $(\pm\frac{L}{2},0)$ respectively as shown in Fig.\ref{fig:BistaticRadarSetup}. We consider a target at ($r_t,\thetat$) where the distance between the target and the transmitter and receiver are $R_{tx}$ and $R_{rx}$ respectively and the bistatic radar angle is $\beta$. We assume that the target is at a fixed bistatic range $\kappa$ with respect to the radar where $\kappa$ is the geometric mean of $R_{tx}$ and $R_{rx}$ ($\kappa = \sqrt{R_{tx}R_{rx}}$). The angular position of the target, $\thetat$, is a uniform random variable between $[0,2\pi)$. Note that the locus of all the points of $\theta_t$ from $[0,2\pi)$ for a fixed $\kappa$ and $L$ forms a Cassini oval  \cite{willis2005bistatic} as shown in Fig.\ref{fig:CassiniOval}. 
\begin{figure}[htbp]
\begin{subfigure}[b]{0.49\textwidth}
         \centering
         \includegraphics[width=0.99\textwidth]{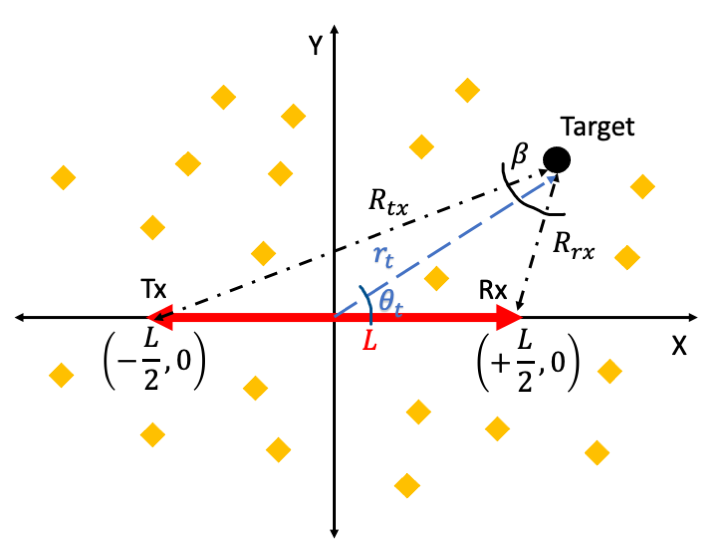}
         \caption{}
         \label{fig:BistaticRadarSetup}    
\end{subfigure}
\hfill
\begin{subfigure}[b]{0.49\textwidth}
         \centering
         \includegraphics[width=0.99\textwidth]{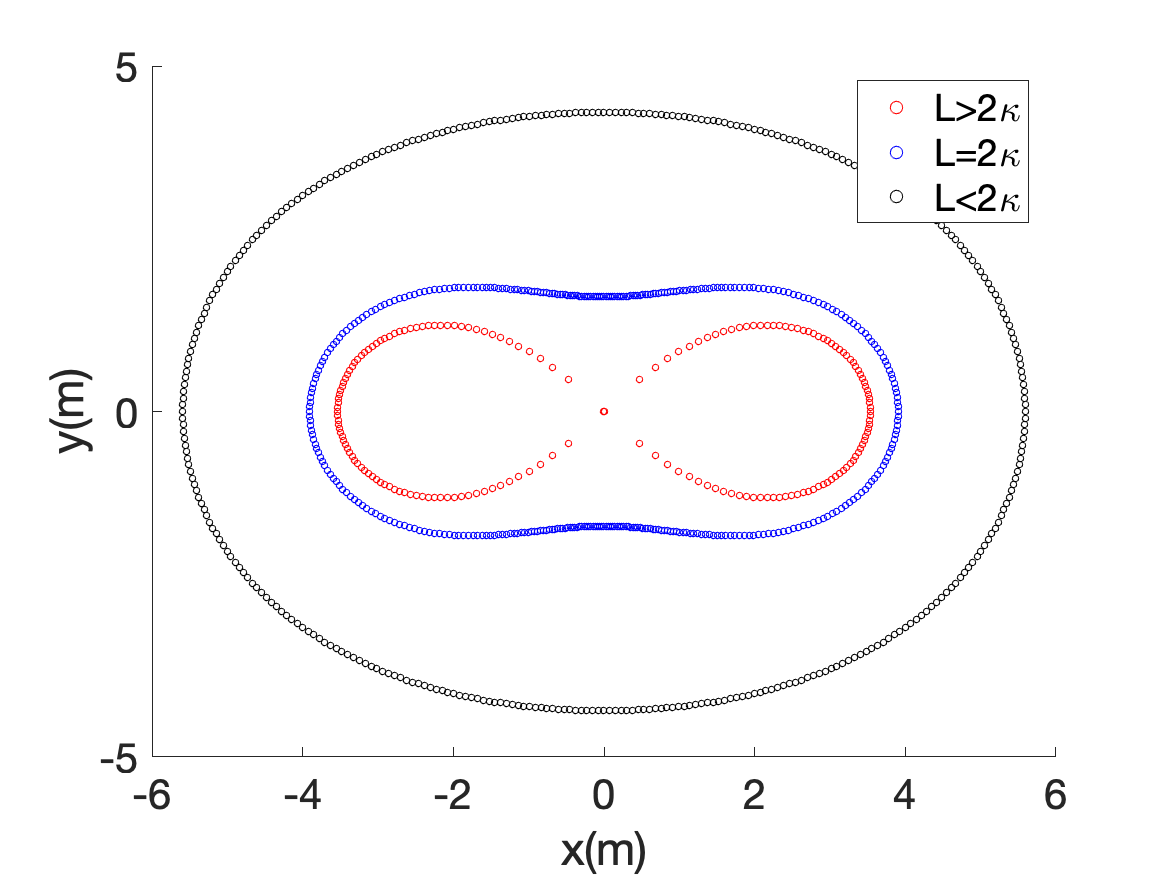}
         \caption{}
         \label{fig:CassiniOval}    
\end{subfigure}
\caption{Histograms of sine of bistatic angle ($\sin \beta$) and minimum normalized one-way range ($R_{min}/\kappa$).}
\vspace{-0.1cm}
\end{figure}
There are three cases: the co-site region ($L < 2\kappa$), the lemiscate region ($L = 2 \kappa$) and when $L > 2\kappa$, the oval breaks into two parts, centered around either transmitter (when $R_{rx}>>R_{tx}$) or the  receiver (when $R_{tx}>>R_{rx}$). We will consider the first two cases in this work. 
For a fixed $\kappa$ and baseline length ($L$) and a specific realization of $\theta_t$, the polar radial position of the target $r_t$ is determined by solving the quadratic equation from the triangle law of cosines given by \par\noindent\small
\begin{align}
\label{eq:SolveForr_t}
    \left(r_t^2+\frac{L^2}{4}\right)^2 - r_t^2L^2\cos^2\theta_t = \kappa^4.
\end{align}
Then we can determine $R_{tx}$, $R_{rx}$ and  $\beta$ from the equations below -\par\noindent\small
\begin{align}
\label{eq:SolveForRt1}
R_{tx}^2 &= \left(r_t^2+\frac{L^2}{4}\right) +r_tL\cos\theta_t\\
\label{eq:SolveForRt2}
R_{rx}^2 &= \left(r_t^2+\frac{L^2}{4}\right) -r_tL\cos\theta_t \\
\label{eq:SolveForRt3}
\cos \beta &= \frac{\left(R_{tx}^2+R_{rx}^2-L^2\right)}{2\kappa^2}.
\end{align}
Some comments regarding the bistatic radar geometry: First, note that when $L=0$, we obtain $R_{tx} = R_{rx} = \kappa$, and the bistatic radar geometry becomes a monostatic case and the Cassini ovals become concentric circles. Second, when $\theta_t = 0$ and for low values of $\kappa$, we obtain the forward scattering scenario when the target is directly between the transmit and receiver. Third, we get the maximum $\beta$ for $\theta_t=\frac{\pi}{2}$ when $R_{tx} = R_{rx}$. Based on equations \eqref{eq:SolveForr_t} to \eqref{eq:SolveForRt3}, we obtain \par\noindent\small
\begin{align}
\label{eq:sinBetaMax}
\sin \beta^{max} = \sqrt{\frac{L^2}{\kappa^2}-\frac{L^4}{4\kappa^4}} 
\end{align}
When we consider the co-site scenario where $L < 2 \kappa$, then the above expression for $\sin \beta^{max} \approx \frac{L}{\kappa}$. On the other hand, in the lemiscate scenario where $L = 2\kappa$, \eqref{eq:sinBetaMax} becomes $\sqrt{3}$. In the third scenario when $L>2\kappa$, we cannot obtain $\beta^{max}$ since the Cassini ovals break into two separate regions as discussed earlier.  
\subsection{Signal to Clutter and Noise Ratio (SCNR)}
We assume that the transmitted power from the radar is $P_{tx}$ and the gains of the transmitting and receiving antennas are $G_{tx}(\theta)$ and $G_{rx}(\theta)$ with half-power beamwidths, $\Delta\theta_{tx}$ and $\Delta \theta_{rx}$ respectively. 
We assume that the target is in the far-field of both the transmitter and receiver. Then the two-way propagation from the bistatic radar transmitter to the target and then back to the receiver is a function of the product of $R_{tx}$ and $R_{rx}$ and indicated as $\hxt$. In the simplest scenario, where there is line-of-sight (LOS) propagation, 
$\hxt$ is \par\noindent\small
\begin{align}
\label{eq:hxt}
    \hxt = \frac{\lambda^2}{(4\pi)^3 R_{tx}^2R_{rx}^2}= \frac{\lambda^2}{(4\pi)^3 \kappa^{4}},
\end{align}
where $\lambda$ is the wavelength corresponding to the carrier frequency.
Then the received power from the target is given by the bistatic radar range equation given by \par\noindent\small
\begin{align}
\label{eq:TgtSignal}
    \textbf{S}(\kappa) = P_{tx}G_{tx}(\thetat)G_{rx}(\thetat)\mathbf{\sigma_t}\hxt,
\end{align}
where $\mathbf{\sigma_t}$ is the bistatic radar cross-section. 
We assume that the target is within the main lobe of both the transmitter and receiver antennas and that the gain pattern is uniform within the lobes. The maximum gains of the transmitting and receiving lobes are inversely proportional to their respective beamwidths $G_{tx}^{max}G_{rx}^{max} = \frac{A_0}{\twotheta}$. The constant of proportionality, $A_0$, is directly related to factors related to the impedance matching efficiency, dielectric and conductor loss and the aperture efficiency and typically determined through empirical studies. Thus, \eqref{eq:TgtSignal} can be written as \par\noindent\small
\begin{align}
\label{eq:TgtSignal2}
    \textbf{S}(\kappa) = P_{tx}G_{tx}^{max}G_{rx}^{max}\mathbf{\sigma_t}\hxt = \frac{P_{tx}A_0\mathbf{\sigma_t}\hxt}{\twotheta}.
\end{align}
Studies have shown that the bistatic $\sigma_t$ follows a similar distribution to monostatic radar cross-section \cite{skolnik1961analysis} which are typically expressed in terms of one of four Swerling models. Specifically, microwave and millimeter wave targets such as humans and vehicles are modeled as Swerling-1 models since they are slow moving with several scattering centers of similar strength \cite{raynal2011radar,Raynal2011RCS}. The Swerling-1 model follows an exponential distribution with a mean radar cross-section $\rcst$ as shown below. \par\noindent\small
\begin{align}
\label{eq:TargetRCS}
  \mathcal{P}(\sigma_t) = \frac{1}{\rcst}exp \left(\frac{-\sigma_t}{\rcst} \right).
\end{align}
In the bistatic radar plane, we consider discrete clutter scatterers, located at $(r_c,\theta_c)$, that represent both direct reflections and multipath from environmental artefacts. For example, in an indoor environment, the target is the human while the multipath from the walls, ceilings and furniture  give rise to ghost targets and clutter. Based on the room geometry and the movement of the human, there can be considerable variation in the strength, spatial distribution and the number of point clutter responses. Similarly, in automotive radar, the multipath reflections from guard railings and surrounding buildings give rise to discrete point clutter responses in the radar data. Again, there is considerable randomness in the strength and spatial distribution of these point clutter responses.
Hence, we model the distribution of these clutter scatterers as a Poisson point process (PPP) - where the number of clutter scatterers in each realization follows a Poisson distribution and the positions of these scatterers follow a uniform distribution. This assumption is consistent with those adopted in previous works \cite{chen2012integrated,ram2020estimating,ram2021optimization}.
The generalized Weibull model \cite{sekine1990weibull} has been frequently used to describe the distribution of the radar cross-section of clutter. Here the distribution is a function of the average bistatic radar cross-section, $\rcsc$ and the shape parameter, $\alpha$ as shown in \par\noindent\small
\begin{align}
\label{eq:ClutterRCS}
 \mathcal{P}(\sigma_c) = \frac{\alpha}{\rcsc}\left(\frac{\sigma_c}{\rcsc}\right)^{\alpha-1}\exp\left(-\left(\frac{\sigma_c}{\rcsc}\right)^{\alpha} \right),
\end{align}
The shape parameter can vary from zero - corresponding to Rayleigh distribution, to one corresponding to exponential distribution and is determined from empirical studies. In scenarios, where the reflectivity of some scatterers dominate the others resulting in strong discrete clutter characteristics, the distribution is most similar to an exponential form \cite{goldstein1973false}. Therefore, in this work we assume that $\mathcal{P}(\sigma_c)$ follows \par\noindent\small
\begin{align}
\label{eq:ClutterRCS2}
 \mathcal{P}(\sigma_c) = \frac{1}{\rcsc}\exp\left(-\frac{\sigma_c}{\rcsc}\right).
\end{align}
We, further, assume that the noise, target and clutter statistics do not change appreciably during the coherent processing interval of the radar. These conditions are generally met for microwave or millimeter-wave radars \cite{billingsley2002low,ruoskanen2003millimeter}. 
In order to determine the clutter returns at the bistatic radar, we specifically consider those clutter scatterers that fall within the same radar resolution cell as the target. There are two types of resolution cells as shown in the figure below: the beamwidth limited resolution cell corresponding to the two-dimensional region that falls within the main lobes of both the transmitting and receiving antennas; and the range limited cell for a radar waveform of $\Delta \tau$ pulse width. The former's cell size for a target at  $(r_t,\theta_t)$ is given by \par\noindent\small
\begin{align}
\label{eq:BWResCell}
    A_{c_{\theta}} = \frac{R_{tx} \Delta \theta_t R_{rx} \Delta \theta_r}{\sin \beta} = \frac{\kappa^2 \twotheta}{\sin \beta}.
\end{align}
For fixed $\kappa$ and antenna beamwidths, the area resolution cell is inversely proportional to $\beta$.
\begin{figure}[htbp]
\begin{subfigure}[b]{0.49\textwidth}
         \centering
         \includegraphics[width=0.99\textwidth]{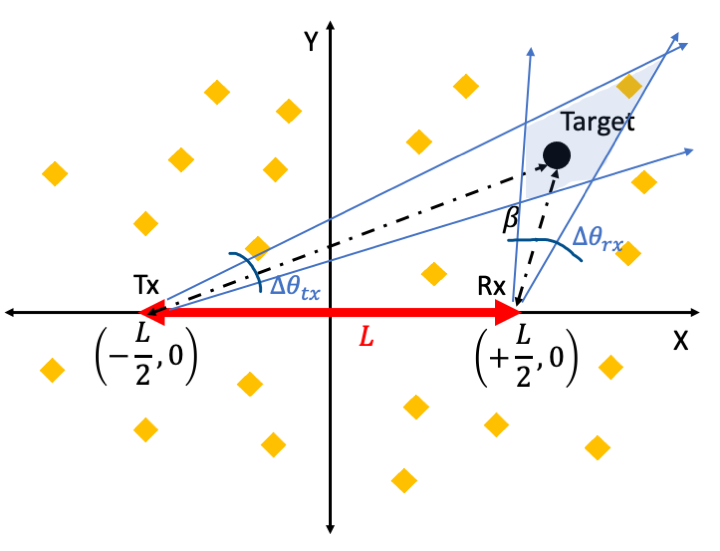}
         \caption{}
         \label{fig:Actheta}    
\end{subfigure}
\hfill
\begin{subfigure}[b]{0.49\textwidth}
         \centering
         \includegraphics[width=0.99\textwidth]{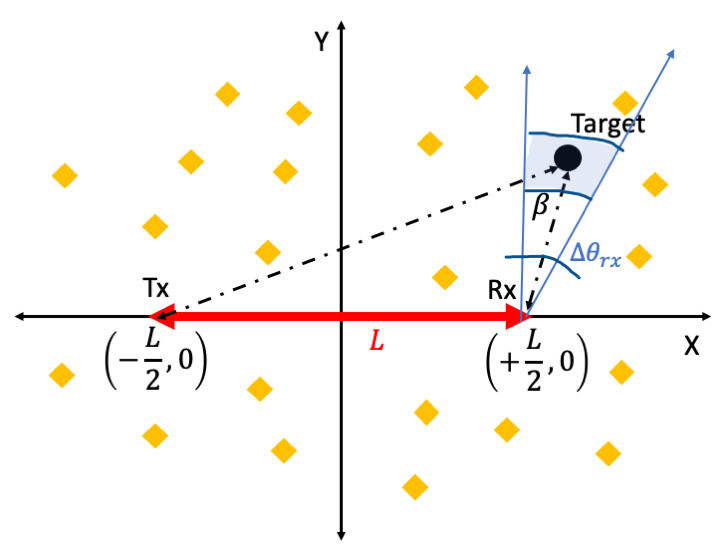}
         \caption{}
         \label{fig:Acr}    
\end{subfigure}
\caption{(a) Beamwidth limited clutter resolution cell ($A_{c_{\theta}}$)and (b) range limited clutter resolution cell ($A_{c_r}$)}
\vspace{-0.1cm}
\end{figure}
In other words, in the limiting scenario when the bistatic radar system resembles a monostatic case (when $\beta = 0$), the beams from both the transmitter and receiver are parallel and hence intersection area or $A_c$ becomes infinite.
The latter's cell size is given by \par\noindent\small
\begin{align}
\label{eq:RangeResCell}
    A_{c_{r}} = \frac{c\Delta \tau R_{min} \Delta \theta_{rx}}{2\cos^2 \beta},
\end{align}
where $R_{min}$ is the minimum of $R_{tx}$ and $R_{rx}$.
All of the clutter scatterers within the same resolution cell as the target give rise to clutter returns, $\mathbf{C}$, at the radar receiver given by \par\noindent\small
\begin{align}
\label{eq:CluttSignal}
    \mathbf{C}(\kappa) = \Su P_{tx}G_{tx}(\thetac)G_{rx}(\thetac)\mathbf{\sigma_c}\hxc.
\end{align}
Here $\theta_c$ and $\kappa_c$ are the polar angle and bistatic range of each clutter scatterer respectively and $\mathbf{A_c}$ is the bistatic radar resolution cell that could correspond to either \eqref{eq:BWResCell} or \eqref{eq:RangeResCell}. Since $\mathbf{\theta_t}$ is a random variable, $\mathbf{A_c}$ which is a function of $\mathbf{\theta_t}$ is also a random variable. For a given noise of the radar receiver, $N_s = K_BT_sBW$ where $K_B, T_s$ and $BW$ are the Boltzmann constant, system noise temperature and bandwidth respectively, the signal to clutter and noise ratio is given by $\mathbf{SCNR}(\kappa) = \frac{\mathbf{S}(\kappa)}{\mathbf{C}(\kappa)+N_s}$
\subsection{Bistatic radar detection coverage probability}
Classical radar detection framework is based on binary hypothesis testing derived from the Neyman-Pearson (NP) theorem.
Here, the basic problem is to decide, on the basis of the radar receiver measurement, between the \emph{alternative hypothesis} ($\mathcal{H}_1$) where the target echo is present along with noise and clutter in the measurement, and the \emph{null hypothesis} ($\mathcal{H}_0$) where only noise and clutter are present in the measurement. Based on the decisions, there are three possible outcomes - a true detection, a missed detection when the signal along with noise and clutter falls below the threshold and a false alarm when the clutter and noise fall above the threshold. The separation between the PDFs of $\mathcal{H}_0$ and $\mathcal{H}_1$ is a function of the mean signal to clutter and noise ratio (SCNR). In other words, when the SCNR is high, it is far easier to obtain a high probability of detection while maintaining a low probability of false alarm.
In this work, we propose a new metric called the radar detection coverage probability that is directly based on the SCNR in the region-of-interest. It is analogous to the radar detection coverage probability that was derived for monostatic radar scenarios in our previous works \cite{ram2020estimating,ram2021optimization}.
\begin{definition}
The bistatic radar detection coverage probability ($\Pdc$) is defined as the probability that the SCNR for a single target at a bistatic range ($\kappa$) is above a predefined threhold, $\gamma$ i.e $\Pdc(\kappa) \triangleq \mathbb{P}\left(\mathbf{SCNR}(\kappa) \geq \gamma \right)$
\end{definition}
\begin{theorem}
The $\Pdc$ for a given bistatic range $\kappa$ where $L < 2\kappa$ corresponding to co-site conditions, with exponential clutter statistics, for a beamlimited radar resolution cell, for a fixed transmitted power and transmitter and receiver beamwdiths, is given by \par\noindent\small
\begin{align}
\label{eq:theorem1}
    \Pdc &= exp\left(-\frac{\gamma N_s\twotheta}{\longvar}-\frac{\rho\kappa^3\twotheta\gamma\rcsc}{L(\rcst+\gamma\rcsc)}\right)
\end{align}
\end{theorem}
\begin{proof}
For a bistatic radar with a target at bistatic range $\kappa$, the $\mathbf{SCNR}$ is a function of several random variables such as the target cross-section, the position of target, the number and spatial distribution of the discrete clutter scatterers and their radar cross-section as shown below. \par\noindent\small
\begin{align}
\label{eq:SCNR1}
\mathbf{SCNR}(\kappa) 
&= \frac{P_{tx}A_0\mathbf{\sigma_t}\hxt/\twotheta}{\Su P_{tx}A_0\mathbf{\sigma_c}\hxc/\twotheta+N_s} 
\\
&= \frac{\mathbf{\sigma_t}}{\Su\frac{\mathbf{\sigma_c}\hxc}{\hxt}+\frac{N_s\twotheta}{P_{tx}A_0\hxt}}
\end{align}
Since $\Pdc = \mathbb{P}\left(\mathbf{SCNR}(\kappa) \geq \gamma\right)$, we can write \eqref{eq:SCNR1} as \par\noindent\small
\begin{align}
\label{eq:Pdc1}
\mathcal{P}\left(\mathbf{\sigma_t}\geq \Su\frac{\gamma\mathbf{\sigma_c}\hxc}{\hxt}+\frac{\gamma N_s\twotheta}{P_{tx}A_0\hxt}\right).
\end{align}
As per \eqref{eq:TargetRCS}, $\mathbf{\sigma_t}$ follows the Swerling-1 model. Hence, \eqref{eq:Pdc1} can be expanded to \par\noindent\small
\begin{align}
\Pdc&=exp \left(\Su\frac{-\gamma\mathbf{\sigma_c}\hxc}{\rcst\hxt}-\frac{\gamma N_s \twotheta}{\longvar} \right) 
\\
\label{eq:ShortPdc}
&=exp \left( \frac{-\gamma N_s \twotheta}{\longvar}\right) I(\kappa).
\end{align}
In the above expression, we have separated the $\Pdc$ into two terms. The first expression which is entirely a function of constants shows the effect of signal to noise ratio on the performance of the bistatic radar, while the second term $I(\kappa)$ shows the effect of signal to clutter ratio on the performance of the system. Since the exponent of sum of terms can be written as a product of exponents, $I(\kappa)$ can be written as \par\noindent\small
\begin{align}
\label{eq:I1}
I = \Ei\left[\Pro exp \left(\frac{-\gamma\mathbf{\sigma_c}}{\rcst} \right)\right],
\end{align}
where $\mathbf{A_c}$, the radar range resolution cell, is a function of $\kappa$ and the position of the target $\mathbf{\theta_t}$. Also, due to the close proximity of the clutter points to the target, we have assumed that $\hxt \approx \hxc$. Based on the principles of stochastic geometry, we use the probability generating functional (PGFL) of a homogeneous PPP \cite{haenggi2012stochastic}, to rewrite, $I(\kappa)$ as \par\noindent\small
\begin{align}
\label{eq:I2}
I &=exp\left(-\underset{\mathbf{\sigma_c},c}{\E}\left[\iint_{\mathbf{r_c,\phi_c}} \rho\left(1- exp\left(\frac{-\gamma\mathbf{\sigma_c}}{\rcst}\right)\right) d(\vec{x}_c) \right]\right) 
\\
&=exp\left(-\Eii\left[\left(1- exp\left(\frac{-\gamma\mathbf{\sigma_c}}{\rcst}\right)\right) \rho \mathbf{A_c}\right]\right).
\end{align}
The integral in \eqref{eq:I2} has to be evaluated in the region of beamwidth limited clutter resolution cell $A_{c_{\theta}}$ as shown in Fig.\ref{fig:Actheta}. If we assume that the clutter scatter statistics are constant within that cell, then the integral can be further reduced to \par\noindent\small
\begin{align}
I= exp\left(-\Eii\left[\left(1- exp\left(\frac{-\gamma\mathbf{\sigma_c}}{\rcst}\right)\right) \rho \frac{\kappa^2\twotheta}{\sin \mathbf{\beta}}\right]\right)
\end{align}
Now the angular position of the target $\mathbf{\theta_t}$ is a uniform random variable from $[0,2\pi)$. We perform a Monte Carlo simulation to study the distribution of the corresponding $\sin \beta$ for 10000 trials of $\theta_t$ for different ratios of $\kappa/L$. The resulting distribution is shown in Fig.\ref{fig:HistogramBeta} where we observe that the $\sin \beta$ takes the value of $\sin \beta^{max}$ with a high probability. Hence, we can further reduce it to \par\noindent\small
\begin{align}
\label{eq:I3}
I &= exp\left(-\underset{\mathbf{\sigma_c}}{\E}\left[1- exp\left(\frac{-\gamma\mathbf{\sigma_c}}{\rcst}\right)\right]  \frac{\rho\kappa^2\twotheta}{\sin\beta^{max}}\right) 
\\
\label{eq:I4}
&= exp\left(-\underset{\mathbf{\sigma_c}}{\E}\left[1- exp\left(\frac{-\gamma\mathbf{\sigma_c}}{\rcst}\right)\right] \Upsilon(\kappa)\right)
\end{align}
Since $\Upsilon(\kappa)$ is a constant independent of $\sigma_c$, it can be pulled out of the integral for computing the expectation as shown below. \par\noindent\small
\begin{align}
\label{eq:Jrc}
I(\kappa)
=exp\left(-\Upsilon(\kappa)\int_0^{\infty}\left(1- exp\left(\frac{-\gamma\mathbf{\sigma_c}}{\rcst}\right)\right)\mathcal{P}(\sigma_c)d\sigma_c \right)
\\
\label{eq:Jrc11}
= exp\left(-\frac{\Upsilon(\kappa)\gamma \rcsc}{\rcst+\gamma \rcsc} \right) = exp\left(-\frac{\rho\kappa^3\twotheta\rcsc}{L(\rcst+\gamma\rcsc)}\right)
\end{align}
The integral in \eqref{eq:Jrc} integrates to 
\eqref{eq:Jrc11}. We substitute $I(\kappa)$ in \eqref{eq:ShortPdc} to prove the theorem in \eqref{eq:theorem1}.
\end{proof}
Now, we discuss the key insights obtained from the theorem. The $\Pdc$ consists of an exponential expression consisting of two parts. The first part shows the effect of the signal to noise ratio while the second part shows the effect of the signal to clutter ratio. 
In other words, when clutter is absent ($\rho =0$) we obtain the performance of the radar solely under noise limited conditions. Similarly when $N_s = 0$, we obtain solely the effect of clutter on the radar performance.\\
\emph{Noise Limited Conditions:}
We will first discuss the performance of the radar under noise-limited conditions. The $\Pdc$ increases with increase in transmitted power, $P_{tx}$; reduced radar antenna beamwidths $\twotheta$; and reduced radar receiver noise. 
Further, as expected, the increase in target's bistatic range causes the performance to deteriorate. In LOS conditions, this deterioration is at the rate of $\kappa^4$. The increase in the mean target cross-section, on the other hand, improves the performance. The findings from this expression are consistent from those obtained from the Frii's bistatic radar range equation and not very dissimilar from the monostatic radar conditions.\\
\emph{Clutter Limited Conditions:}
Next, we discuss the second term where we study how the presence of discrete clutter scatterers affects the performance of the radar.
First, we note that, unlike the noise limited conditions, the performance is independent of the transmitted power since an increase in $P_{tx}$ causes a proportionate increase in both target and clutter returns. Second, the reduction in the radar antenna beamwidths, $\twotheta$, improves the detection performance in a similar manner to the noise limited conditions. This is because the clutter resolution cell size reduces for lower beamwidths resulting in lower returns. 
In LOS conditions, the $\Pdc$ deteriorates at the rate of $\kappa^3$ instead of $\kappa^4$ (as is the case of noise limited conditions). This is similar to the monostatic radar scenarios where the signal to clutter ratios deteriorated at the rate of the third power of the monostatic range. 
Next, the farther the target is from the radar, the effect of the clutter returns begin to dominate over the effect of noise due to the increase in the clutter resolution cell size. 
\begin{corollary}
\label{corr:corr1}
In LOS conditions, the bistatic range, $\tilde{\kappa}$, at which the transition between the noise limited radar performance to clutter limited radar performance occurs at
\begin{align}
\label{eq:corr1}
    \tilde{\kappa}= \frac{\rho \rcsc \rcst P_{tx}A_0\lambda^2}{L(\rcst+\gamma\rcsc)N_s}.
\end{align}
We prove the above corollary by equating the two terms within the exponent in \eqref{eq:theorem1} at $\tilde{\kappa}$. Further, we assume LOS propagation factor for $\hxt$, from \eqref{eq:hxt}, as shown below.
\begin{align}
    \frac{\gamma N_s\twotheta (4\pi)^3 \tilde{\kappa}^4 }{\rcst P_{tx}A_0 \lambda^2}=\frac{\rho \tilde{\kappa}^3 \twotheta \rcsc}{L(\rcst+\gamma\rcsc)}
\end{align}
\end{corollary}
\begin{corollary}
\label{corr:corr2}
Converse to the previous corollary, the maximum power beyond which there will be no further improvement in the radar detection performance due to the transition from noise limited to clutter limited conditions for a fixed bistatic range is
\begin{align}
\label{eq:corr2}
    P_{tx}^{max} = \frac{L\kappa(\rcst+\gamma\rcsc)N_s}{\rho \rcsc \rcst A_0\lambda^2}
\end{align}
\end{corollary}
As the ratio $\kappa/L$ increases, the bistatic radar scenario begins to strongly resemble the monostatic radar scenario. As a result the resolution cell becomes larger. In fact, in the limiting scenario when $\kappa/L$ tends to infinity, the beams of the transmitter and receiver antennas are parallel and the resolution cell is infinite in size. At this point, it is more beneficial to consider other types of clutter resolution cells instead of the beamwidth resolution cell. In current literature, there is no specific value of $\kappa/L$ to mark the transition from bistatic radar to monostatic radar. 
\begin{corollary}
\label{corr:corr3}
In the absence of noise, the bistatic range, $\overline{\kappa}$, at which the transition from bistatic to monostatic radar scenario occurs when 
\begin{align}
\label{eq:corr3}
    \overline{\kappa} = \left(\frac{L(\rcst+\gamma\rcsc)}{\rho\twotheta\rcsc} \right)^{1/3}.
\end{align}
We show the above corollary by identifying $\overline{\kappa}$ as the value at which $\Pdc$ reduces to 36.7\% ($e^{-1}$) of the maximum possible value as shown below
\begin{align}
exp\left(-\frac{\rho(\overline{\kappa})^3\twotheta\rcsc}{L(\rcst+\gamma\rcsc)}\right) = exp(-1)
\end{align}
\end{corollary}
Increase in the density of the clutter returns, $\rho$, gives causes a deterioration in $\Pdc$ and hence the radar detection performance. 
However, the same is not true for the mean clutter cross-section $\rcsc$.
\begin{corollary}
\label{corr:corr4}
For large values of $\rcsc$, the radar detection performance metric, $\Pdc$, is independent of the clutter cross-section as shown below.
\begin{align}
\label{eq:corr4}
    \lim_{\rcsc\to\infty} exp\left(-\frac{\rho\kappa^3\twotheta\rcsc}{L(\rcst+\gamma\rcsc)}\right) = exp\left(-\frac{\rho\kappa^3\twotheta}{L\gamma}\right)
\end{align}
\end{corollary}
\begin{theorem}
The probability of detection coverage for a given bistatic range $\kappa$ where $L = 2\kappa$ corresponding to lemiscate conditions, with exponential clutter statistics, for a beamlimited radar resolution cell, for a fixed transmitted power and transmitter and receiver beamwdiths, is given by \par\noindent\small
\begin{align}
\label{eq:theorem2}
    \Pdc &= exp\left(-\frac{\gamma N_s\twotheta}{\longvar}-\frac{\rho\kappa^2\twotheta\gamma\rcsc}{\sqrt{3}(\rcst+\gamma\rcsc)}\right)
\end{align}
\end{theorem}
\begin{proof}
 The proof for this theorem follows the identical steps of the previous theorem from \eqref{eq:SCNR1} to \eqref{eq:I4}. In the final step, since $\sin \beta^{max} = \sqrt{3}$, we obtain
 \begin{align}
 \label{eq:Ups2}
     \Upsilon(\kappa)=\frac{\rho \kappa^2\twotheta\gamma\rcsc}{\sqrt{3}(\rcst+\gamma\rcsc)}
 \end{align}
 Substituting \eqref{eq:Ups2} in \eqref{eq:I4}, we prove \eqref{eq:theorem2}.
\end{proof}
The performances of the radar under noise limited and clutter limited conditions in lemiscate scenario do not differ significantly from the co-site scenario. Hence, we do not include a separate discussion on the effect of radar, target and clutter parameters on the detection performance.  
\begin{theorem}
The probability of detection coverage for a given bistatic range $\kappa$ and antenna beamwidth $\Delta \theta_{tx}=\Delta \theta_{rx} = \Delta \Theta$ and $L \leq 2\kappa$ (co-site/lemiscate regions) for a  range limited radar resolution cell is given by 
\begin{align}
\label{eq:theorem3}
\begin{split}
\Pdc = exp \left( \frac{-\gamma N_s\twotheta}{\longvar}\right) \\exp\left(-\frac{\rho c \Delta \tau \Delta \Theta \kappa^2\rcsc}{2(\kappa_m+\sqrt{\kappa_m^2-L^2})(\rcst+\gamma \rcsc)}\right)
\end{split}
\end{align}
\end{theorem}
\begin{proof}
We begin with similar steps to the proof of the first theorem by considering steps \eqref{eq:SCNR1} to \eqref{eq:I2}. Instead of considering the beamwidth limited resolution cell, we consider the range limited resolution cell given in \eqref{eq:RangeResCell}. If we assume that the radar beamwidths for the transmitter and receiver are equal $\Delta \theta_{tx} = \Delta \theta_{rx} = \Delta\Theta$, then $A_{c_r}$ is a function of the minimum one-way range, $R_{min}:min(R_{tx},R_{rx})$. In order to find $R_{min}$, we perform a Monte Carlo simulation with 2000 trials where the polar angle $\theta_t$ is drawn from a uniform distribution from $[0,2\pi)$ and compute the corresponding $R_{min}$ normalized by $\kappa$. The resulting distribution of the normalized $R_{min}$ is presented in Fig.\ref{fig:HistogramRmin}. 
\begin{figure}[htbp]
\begin{subfigure}[b]{0.49\textwidth}
         \centering
         \includegraphics[width=\textwidth]{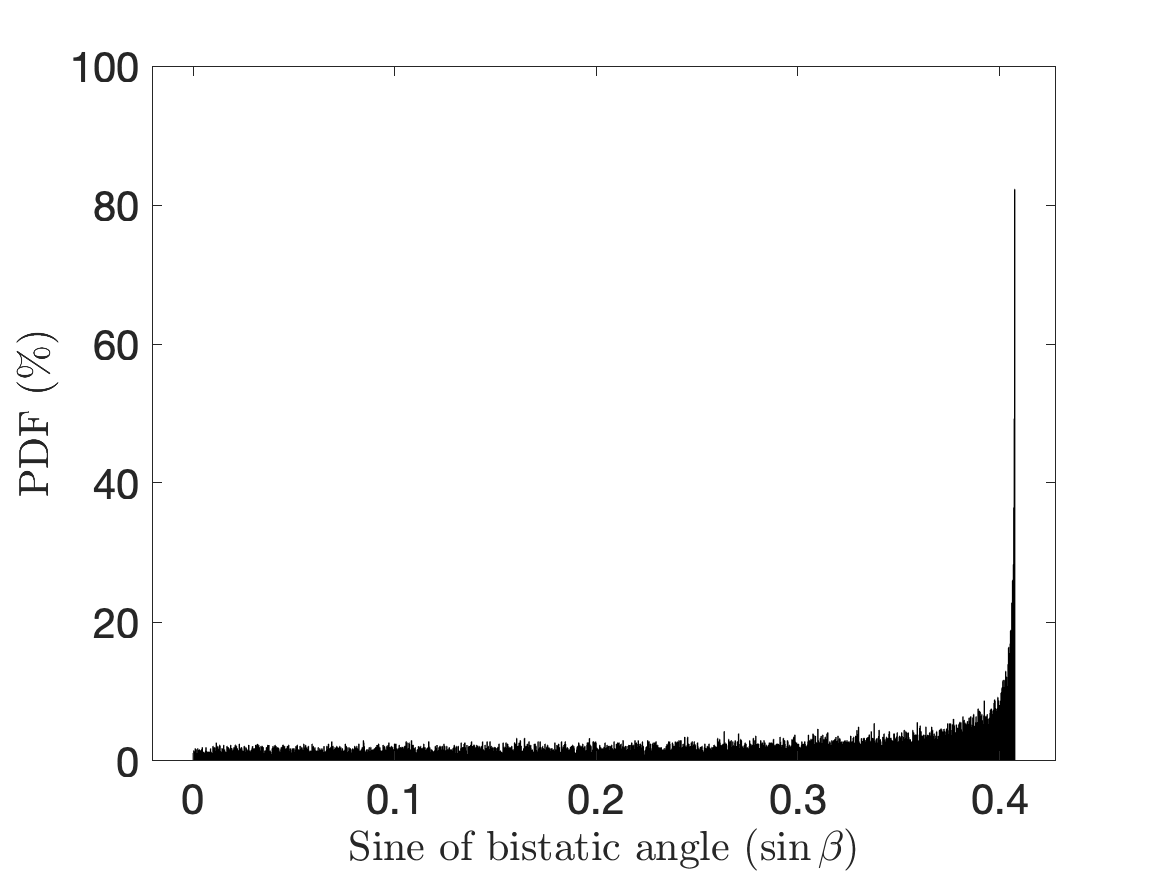}
         \caption{}
         \label{fig:HistogramBeta}    
\end{subfigure}
\hfill
\begin{subfigure}[b]{0.49\textwidth}
         \centering
         \includegraphics[width=\textwidth]{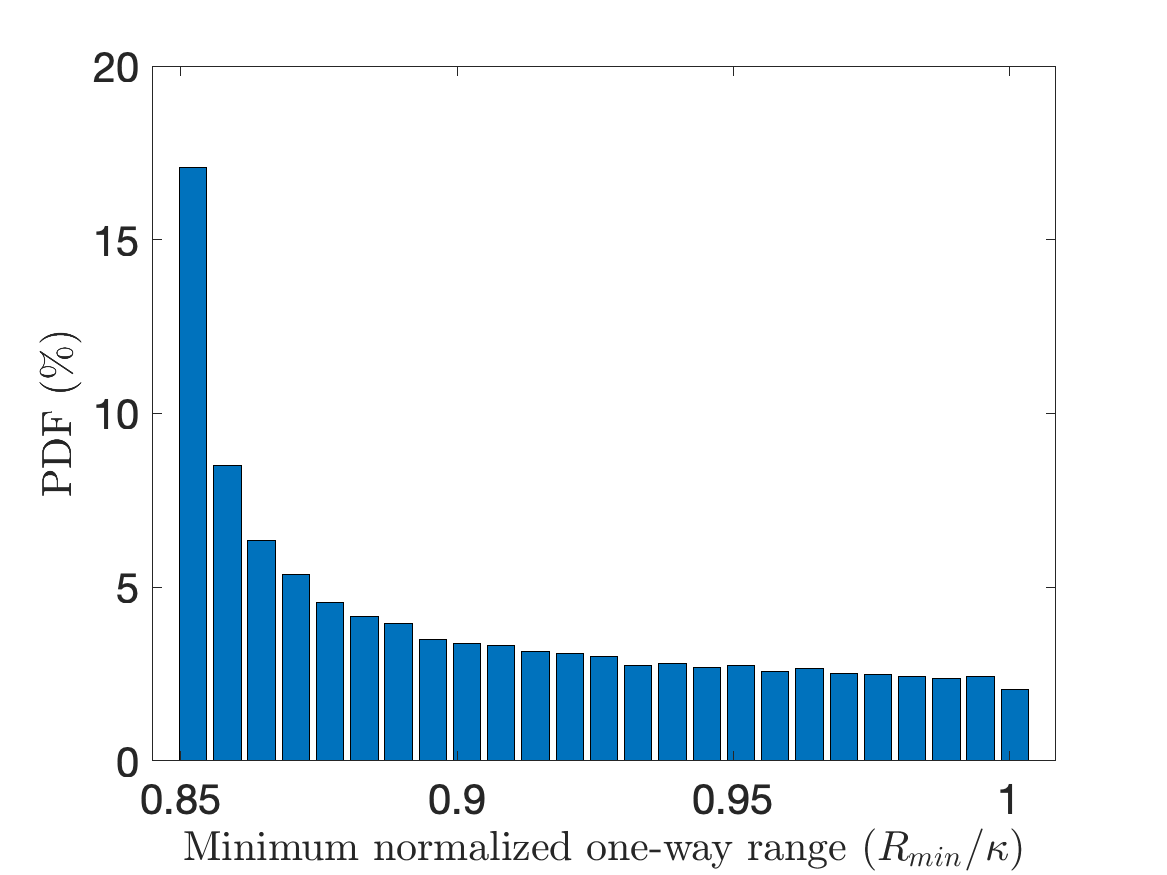}
         \caption{}
         \label{fig:HistogramRmin}    
\end{subfigure}
\caption{Histograms of sine of bistatic angle ($\sin \beta$) and minimum normalized one-way range ($R_{min}/\kappa$).}
\vspace{-0.1cm}
\end{figure}
The simulation result shows that $R_{min}$ varies very closely with respect to $\kappa$ ranging from 0.95 to 1. Hence, without loss of accuracy, we use $\kappa$ in the expression for $A_{c_r}$. Further based on the previous study of the distribution of $\sin\beta$ presented in Fig.\ref{fig:HistogramBeta}, we use $\cos\beta^{max}$ also for $A_{c_r}$. The resulting expression for $A_{c_r}$ is given by
\begin{align}
\label{eq:Acr}
A_{c_r} = \frac{c\Delta \tau \Delta \Theta \kappa^2}{2(\kappa_m+\sqrt{\kappa_m^2-L^2})}.
\end{align}
Substituting this in \eqref{eq:I2}, we obtain
\begin{align}
\label{eq:I5}
I = exp\left(\underset{\mathbf{\sigma_c}}{\E} \left[1-exp\left(\frac{-\gamma\mathbf{\sigma_c}}{\rcst}\right)\frac{\rho c\Delta \tau \Delta \Theta \kappa^2}{2(\kappa_m+\sqrt{\kappa_m^2-L^2})}\right]\right)
\end{align}
For an exponential distribution of $\sigma_c$, we obtain
\begin{align}
\label{eq:JRC}
I = exp\left(-\frac{\rho c \Delta \tau \Delta \Theta \kappa^2\rcsc}{2(\kappa_m+\sqrt{\kappa_m^2-L^2})(\rcst+\gamma \rcsc)}\right).
\end{align}
Substituting this back in \eqref{eq:ShortPdc}, we prove \eqref{eq:theorem3}.
\end{proof}
Similar to the corollaries drawn from Theorem 1, we can derive $\tilde{\kappa}$ where the radar detection scenario transitions from noise limited conditions to clutter limited conditions; and the $P_{tx}^{max}$ which is the maximum power beyond which there will be no further improvement in the radar performance due to high clutter returns. Also, as shown previously in Corollary 1.4, when $\rcsc$ is high, there is no further deterioration in the performance of the radar. Instead, the detection is affected more by the clutter density $\rho$ . However, there is a unique inference for conditions of the radar bandwidth in the range limited scenario. For high values of radar bandwidth, the noise power at the radar receiver increases due to  $N_s = K_B T_sBW^{max}$. However, the clutter returns reduce since $A_{c_r}$ is a function of the pulse width $\Delta \tau$ which is inversely proportional to bandwidth. 
\begin{corollary}
\label{corr:corr5}
The radar bandwidth for optimal detection performance for an impulse radar with a pulse width of $\Delta \tau = 1/BW^{max}$ and noise power is given by
\begin{align}
\label{eq:corr5}
BW^{max} = \left(\frac{\rho c \rcsc\rcst P_{tx}A_0 \hxt \kappa}{2(1-\frac{L^2}{\kappa^2})(\rcst+\gamma\rcsc)\Delta \Theta K_BT_s}\right)^{1/2}
\end{align}
\end{corollary}
The corollary is proved by equating the first order derivative of $\Pdc$ in \eqref{eq:theorem3} with respect to $BW$ to 0.
\section{Results}
\label{sec:Results}
In this section, we present the results from the theorem and its corollaries. They are validated with Monte Carlo simulations. The experimental set up for the Monte Carlo simulations consist of a two-dimensional Cartesian space spanning $X:-100m:+100m$ and $Y:-100m:+100m$ with the bistatic radar located at $(\pm 2.5m,0)$. For each realization of the Monte Carlo simulation, we assume a single target to be located at bistatic  range $\kappa$ at a random realization of $\theta_t$ drawn from a uniform distribution between $[0,2\pi)$ and a radar cross-section drawn from the exponential distribution with a mean $\rcst$. Using the bistatic radar geometry concepts presented in Section \ref{sec:Theory}, we compute $R_{tx}, R_{rx}$ and $\beta$. Then for a given transmitter power, radar antenna beamwidths and wavelength, we compute the received power at the radar for every realization assuming LOS conditions using \eqref{eq:TgtSignal}. Then we model the discrete clutter scatterers as a PPP - where for each realization of the Monte Carlo simulation, the number of clutter scatterers are drawn from a Poisson distribution of mean ($\rho\times$ total area of the Cartesian space) and the position is a uniform random variable in the two-dimensional space. The cross-section of each clutter scatterer is drawn from an exponential distribution of mean $\rcsc$. We determine if the scatterer is within the beamwidth limited resolution cell and then compute the total clutter returns using \eqref{eq:CluttSignal}.
Then for a given system noise temperature and bandwidth, we compute the noise power of the radar system and the SCNR of each realization. We determine the $\Pdc$ by summing the total number of realizations for which the SCNR is above the threshold. The radar, target and clutter parameters that have been used in the simulations are summarized in Table.\ref{tab:param}.
\begin{table}[htbp]
    \centering
    \caption{\small Radar, target and clutter parameters}
    \begin{tabular}{p{1.5cm}|p{1cm}||p{1.5cm}|p{1.5cm}}
    \hline \hline
         Parameter & Values & Parameter & Values\\
    \hline \hline
         $L$ & 5m & $\rcst$ & $1m^2$ \\
         $P_{tx}$ & 10W & 
         $\theta_t$ & $\mathcal{U}[0,2\pi)$ \\ 
         $\rcsc$ & $1m^2$ & $T_s$ & 300K \\
         $\Delta \Theta$ & $5^{\circ}$ & $BW$ & 2GHz\\
         $\lambda$ & 5mm & $\rho$ & $0.001/m^2$\\
        \hline \hline
    \end{tabular}
    \label{tab:param}
    \vspace{-0.1cm}
\end{table}
Figure.\ref{fig:PdcvsKappa_NLCLNC} shows the natural logarithm of $\Pdc$ as a function of $\kappa$ under noise limited and clutter limited conditions. In the noise limited conditions, $\rho$ is set to zero while $T_s$ is set at 300K. We observe that $ln(\Pdc)$ decays at the rate of the fourth power of $\kappa$. On the other hand, when $T_s$ is set at 0 and $\rho$ is $0.001/m^2$, then $ln(\Pdc)$ decays at the rate of the third power of $\kappa$. When both noise and clutter conditions exist, then the performance follows the clutter response till $\tilde{\kappa}$ after which it follows the noise response. We use corollary \ref{corr:corr1} to compute $\tilde{\kappa}$ and it is verified in the figure. 
\begin{figure}[ht]
     \begin{subfigure}[b]{0.49\textwidth}
         \centering
         \includegraphics[width=2.5in,height=2in]{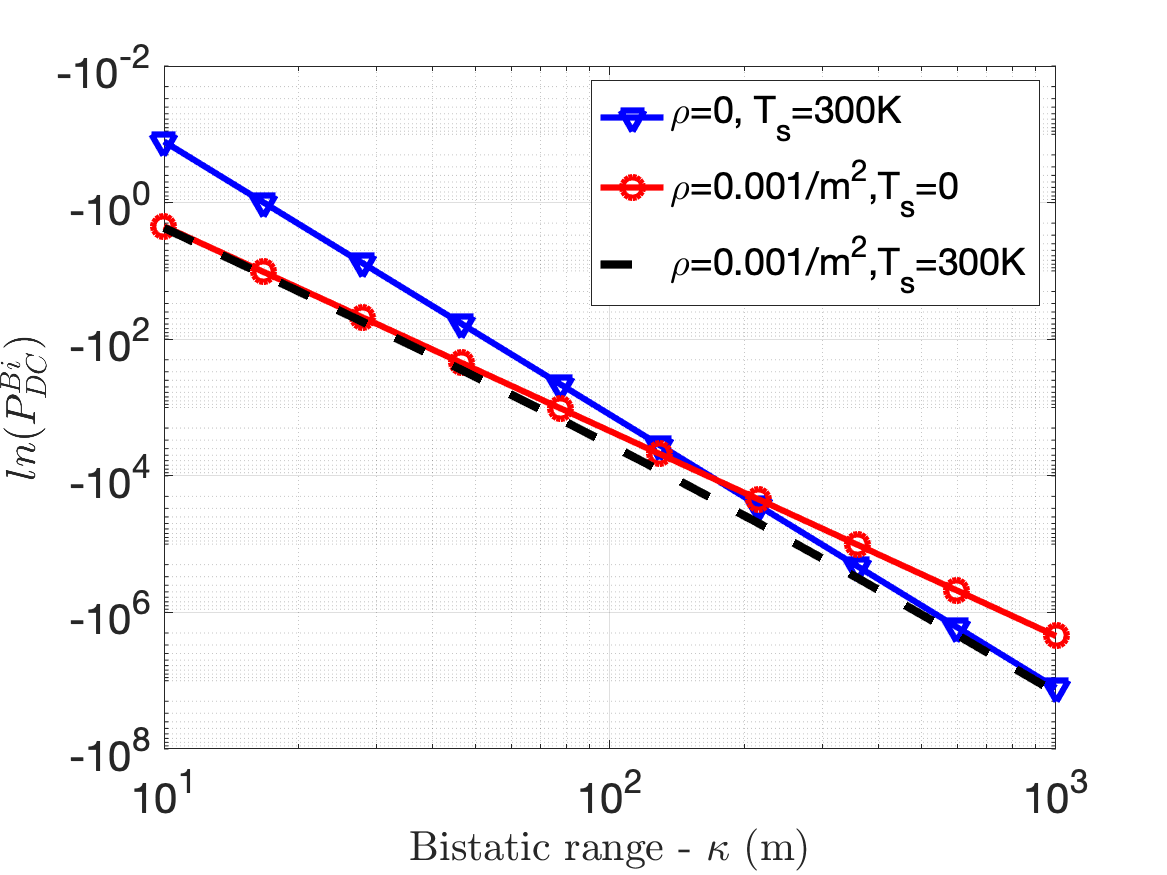}
         \caption{$ln(\Pdc)$ vs. $ln(\kappa)$ for varying $\rho$}
         \label{fig:PdcvsKappa_NLCLNC}
     \end{subfigure}
     \hfill
     \begin{subfigure}[b]{0.49\textwidth}
         \centering
         \includegraphics[width=2.5in,height=2in]{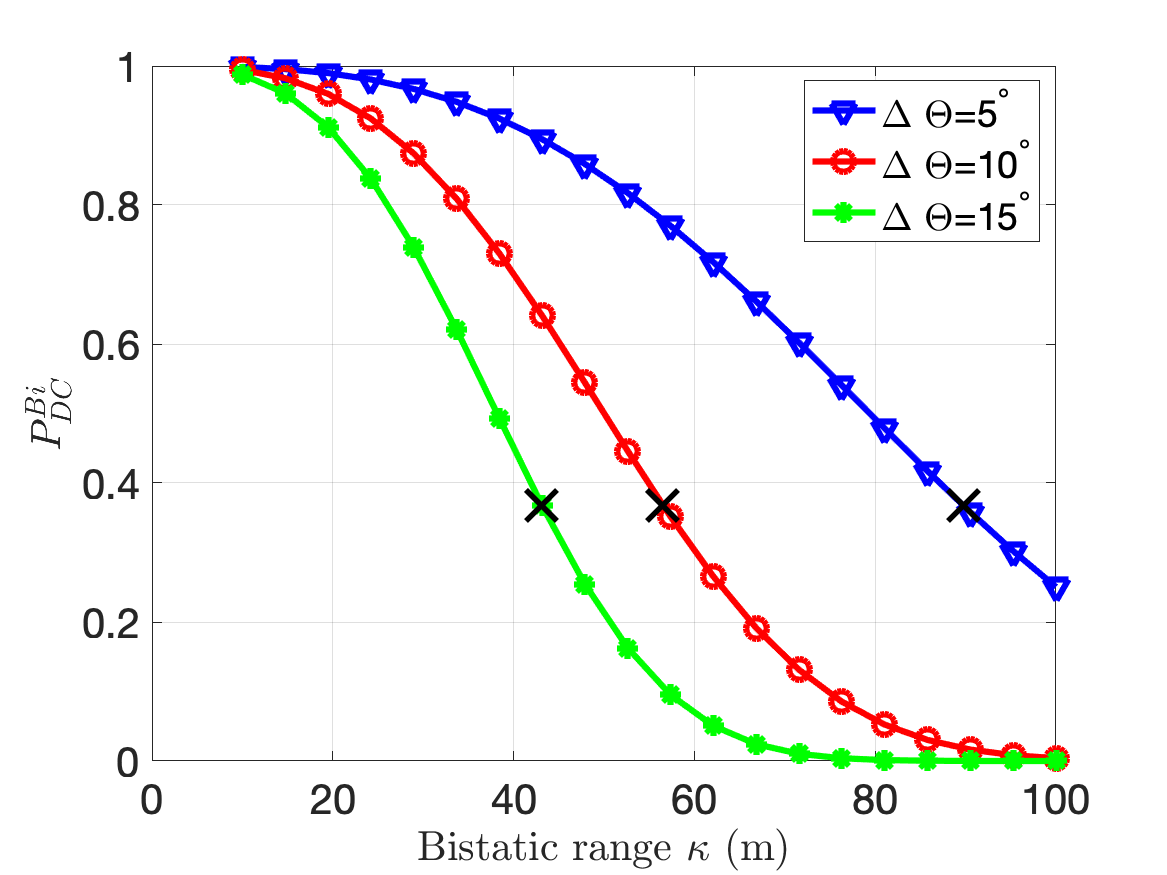}
         \caption{$\Pdc$ vs. $\kappa$ for varying $\Delta \Theta$ }
         \label{fig:PdcvsKappa_Beamwidth}
     \end{subfigure}
     \\
     \begin{subfigure}[b]{0.49\textwidth}
         \centering
         \includegraphics[width=2.5in,height=2in]{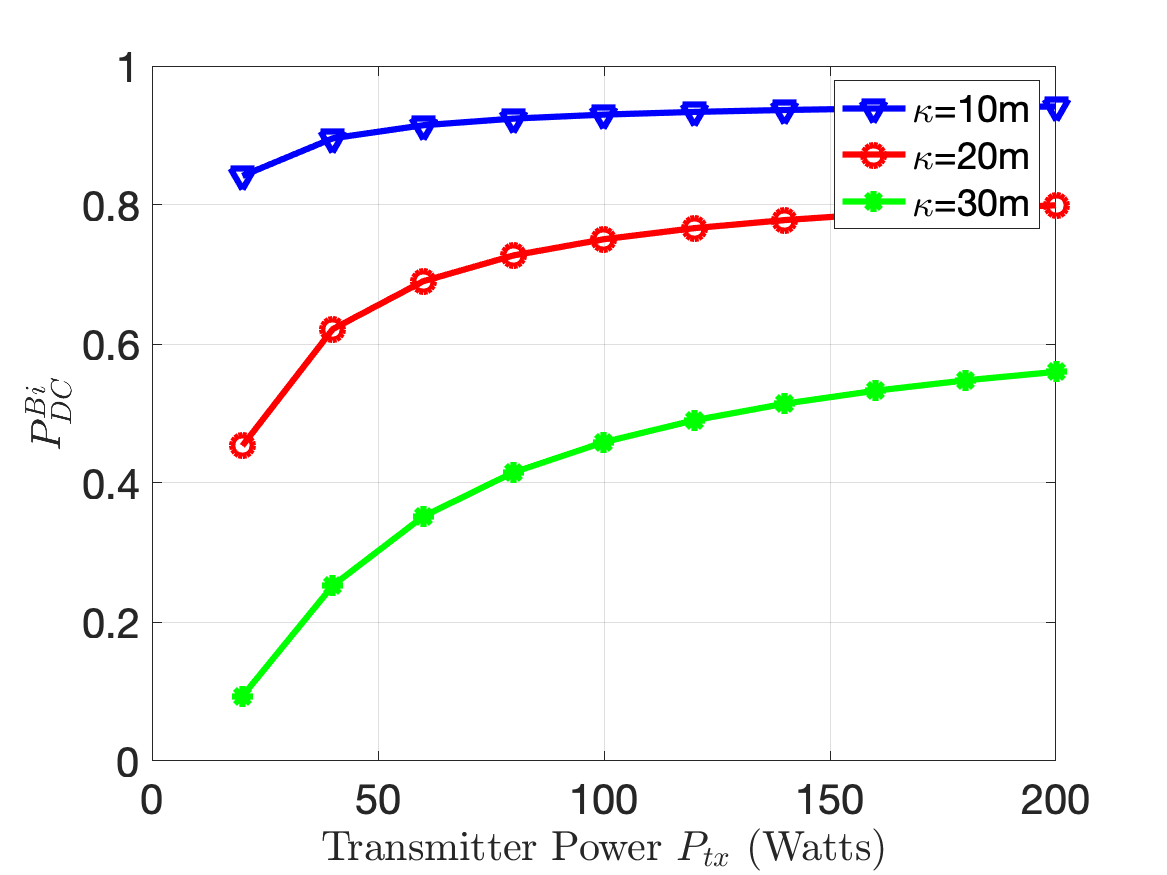}
         \caption{$\Pdc$ v. $P_{tx}$ for varying $\kappa$}
         \label{fig:PdcvsPtx_Kappa}
     \end{subfigure}
     \hfill
          \begin{subfigure}[b]{0.49\textwidth}
         \centering
         \includegraphics[width=2.5in,height=2in]{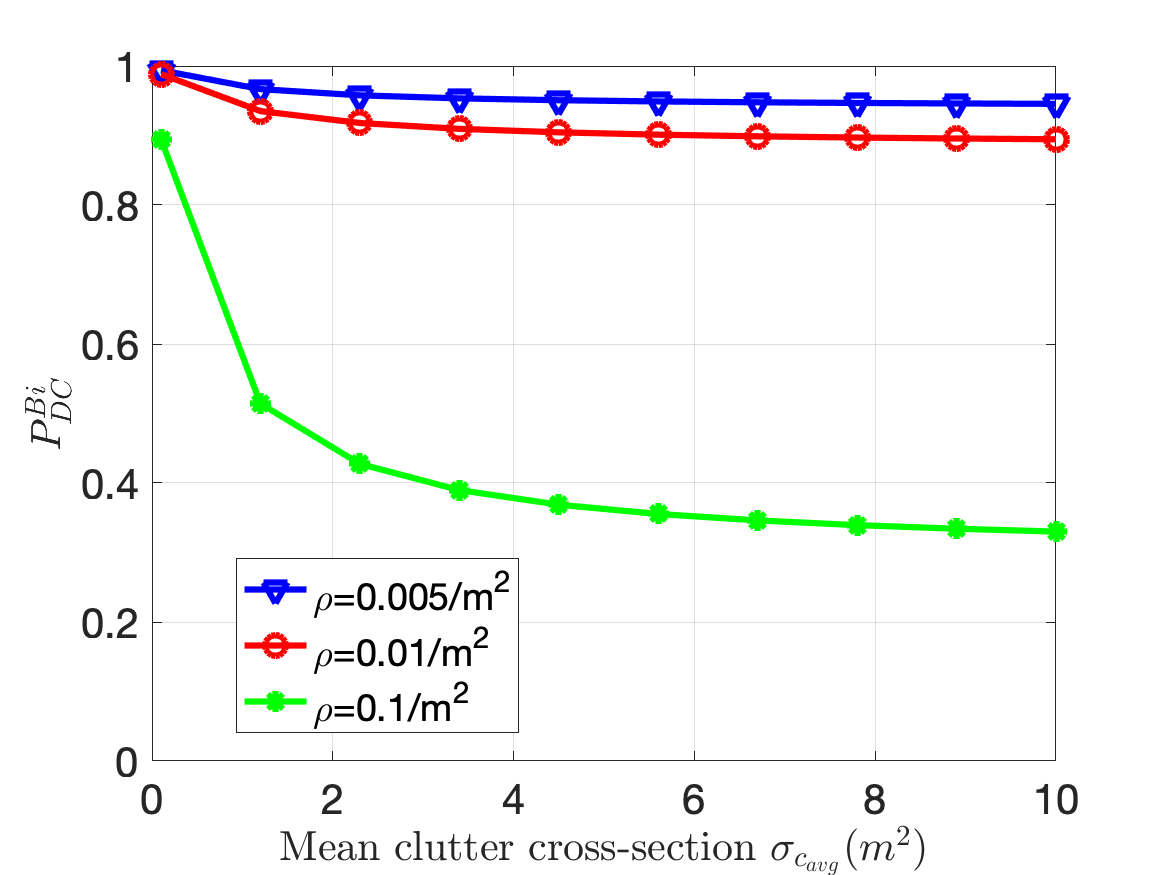}
         \caption{$\Pdc$ vs. $\rcsc$ for varying $\rho$}
         \label{fig:PdcvsSigmac_Rho}
     \end{subfigure}
     \\
     \begin{subfigure}[b]{0.49\textwidth}
         \centering
         \includegraphics[width=2.5in,height=2in]{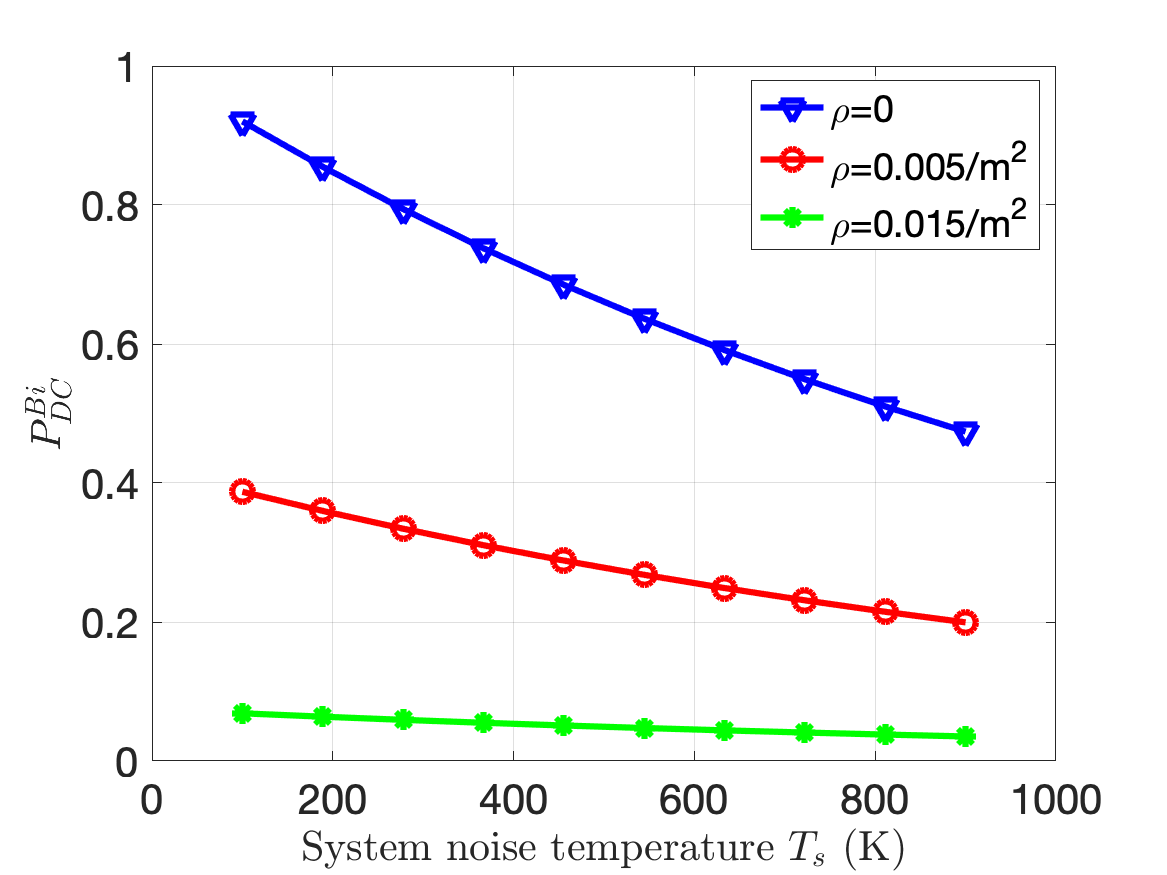}
         \caption{$\Pdc$ vs. $T_s$ for varying $\rho$}
         \label{fig:PdcvsTs_Rho}
     \end{subfigure}
     \hfill
          \begin{subfigure}[b]{0.49\textwidth}
         \centering
         \includegraphics[width=2.5in,height=2in]{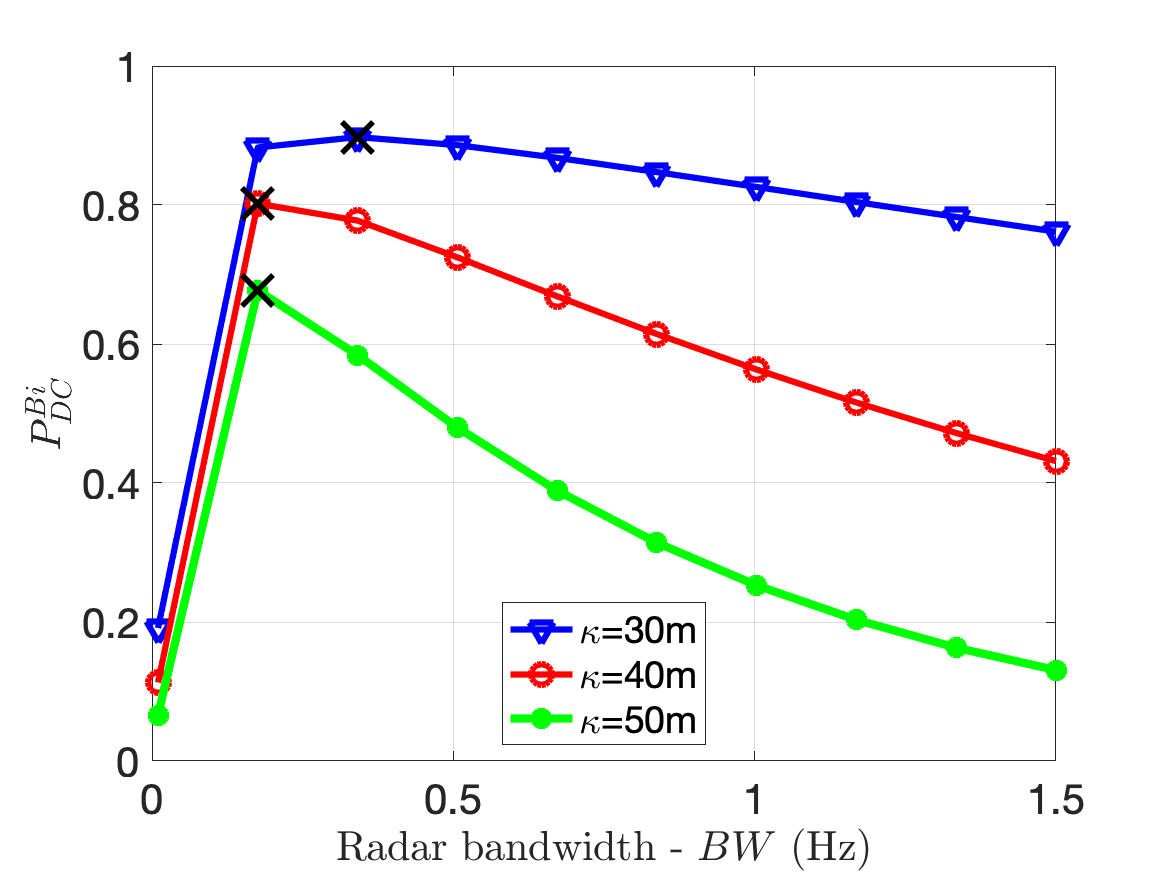}
         \caption{$\Pdc$ vs. $BW$ for varying $\kappa$}
         \label{fig:PdcvsBW_Kappa}
     \end{subfigure}
     \hfill
        \caption{Bistatic radar detection coverage probability ($\Pdc$) as a function of radar, target and clutter parameters.}
        \label{fig:three graphs}
        \vspace{-0.5cm}
\end{figure}
Next, we consider the effect of $P_{tx}$ on the radar detection performance in Fig.\ref{fig:PdcvsPtx_Kappa}. We compute the $\Pdc$ for three different values of $\kappa$. For each case, we observe that the $\Pdc$ improves with increase in $P_{tx}$ under the noise limited conditions. However, beyond a certain maximum value of the $P_{tx}$, there is no further significant improvement in the detection performance due to the high clutter returns. The maximum value of the $P_{tx}$ at which this occurs is computed from corollary \ref{corr:corr2} and corroborates the results shown in the figure.
Next, we study the effect of the clutter cross-section and the clutter density. We fix $\rcst$ at $1m^2$ while $\rcsc$ is varied from $0.1m^2$ to $10m^2$. We study the variation for three different values of $\rho$. We observe that the detection performance improves with increase in $\rcsc$ till $\rcsc>\rcst$. This is consistent with corollary \ref{corr:corr4}.
However, the performance deteriorates significantly with increase in $\rho$. In other words, the radar performance is greatly affected by the \emph{number} rather than the \emph{strength} of the clutter scatterers. 
Next, we examine the impact of the radar performance for different antenna beamwidths in Fig.\ref{fig:PdcvsKappa_Beamwidth}. We assume that beamwidth for both transmitter and receiver antennas are equal to $\Delta \Theta$.
We observe the $\Pdc$ fall with increase in $\kappa$ and for increase in $\Delta \Theta$. The increase in radar gain is beneficial in both the noise limited and clutter limited scenarios. In the noise limited scenario, the increase in gain results in stronger returns from the target. In the clutter limited scenario, the reduced beamwidth results in smaller clutter resolution cells resulting in weaker clutter returns. The view graph also indicates the values of $\overline{\kappa}$ at which the $\Pdc$ falls to 36.7\% of the maximum value corroborating with corollary \ref{corr:corr3}. We can assume that at this bistatic range, the radar system can be approximated to monostatic system due to high $\kappa/L$.
Finally, we study the effect of $T_s$ on the radar detection performance for different values of $\rho$. As anticipated, the detection performance falls for increase in $T_s$ and increase in $\rho$. However, the slopes of the three curves in the view graph also indicates that for greater values of $\rho$, the change in the radar detection performance is less significantly impacted by increase in $T_s$. This is because for greater clutter, the radar operates in the clutter limited scenario and is independent of $T_s$.
All the results that have been discussed so far are for the beamwidth limited clutter. Several of these trends also apply to the range limited clutter and hence not presented due to space constraints. Instead, we show the impact of radar bandwidth on the detection performance in Fig.\ref{fig:PdcvsBW_Kappa} for different values of $\kappa$. The result shows that the performance improves with increase in $BW$ due to reduction in $A_{c_r}$ till it reaches a maximum after which the detection performance deteriorates due to increase in noise. The optimal bandwidth is indicated in the view graph and corroborates with corollary \ref{corr:corr5}.
\section{Conclusion}
\label{sec:Conclusion}
We use stochastic geometry formulations to analyze bistatic radar detection performance through a metric - $\Pdc$. This metric quantifies the average area of the radar field-of-view over which the SCNR is above a predefined threshold. It is a function of radar parameters such as transmitted power, antenna gains, pulse width and noise as well as target and clutter parameters. Through this metric, we obtain key system insights including the bistatic range at which the radar transitions from noise to clutter limited conditions; and approximates to monostatic behavior; as well as the maximum transmitted power and bandwidth for realizing peak detection performance. Due to the flexibility afforded by the analytical framework, it can be used for The results are validated with Monte Carlo simulations. 
\bibliographystyle{ieeetran}
\bibliography{main}

\end{document}